\documentclass[twocolumn,prl,aps,longbibliography,superscriptaddress,showkeys,showpacs]{revtex4-2}
\usepackage[bookmarks=false,linkcolor=blue,urlcolor=blue,colorlinks,citecolor=blue]{hyperref}
\usepackage{amsmath, amsthm, amssymb, mathtools, graphicx, caption, xcolor}

\usepackage{xcolor}

\newtheorem*{theorem*}{Theorem}
\newtheorem*{remark*}{Remark}

\newcommand{ \bs }{ \boldsymbol }
\newcommand{ \myind }{ \overline{\alpha} }

\makeatletter
\def\moverlay{\mathpalette\mov@rlay}
\def\mov@rlay#1#2{\leavevmode\vtop{%
   \baselineskip\z@skip \lineskiplimit-\maxdimen
   \ialign{\hfil$\m@th#1##$\hfil\cr#2\crcr}}}
\newcommand{\charfusion}[3][\mathord]{
    #1{\ifx#1\mathop\vphantom{#2}\fi
        \mathpalette\mov@rlay{#2\cr#3}
      }
    \ifx#1\mathop\expandafter\displaylimits\fi}
\makeatother

\newcommand{\cupdot}{\charfusion[\mathbin]{\cup}{\cdot}}
\newcommand{\bigcupdot}{\charfusion[\mathop]{\bigcup}{\cdot}}

\theoremstyle{plain}
\newtheorem{lemma}{Lemma}
\newtheorem{theorem}[lemma]{Theorem}
\newtheorem{remark}[lemma]{Remark}

\begin{document}

\author{Mladen Kova\v{c}evi\'{c}}
\email{kmladen@uns.ac.rs\\}
\affiliation{University of Novi Sad, Serbia}

\title{Signaling to Relativistic Observers:\\\vspace{1mm}An Einstein--Shannon--Riemann Encounter}

\date{October 15, 2020}

\begin{abstract}
A communication scenario is described involving a series of events triggered
by a transmitter and observed by a receiver experiencing relativistic time
dilation.
The message selected by the transmitter is assumed to be encoded in the events'
timings and is required to be perfectly recovered by the receiver, regardless
of the difference in clock rates in the two frames of reference.
It is shown that the largest proportion of the space of all $ k $-event signals
that can be selected as a code ensuring error-free information transfer in this
setting equals $ \zeta(k)^{-1} $, where $ \zeta $ is the Riemann zeta function.
\end{abstract}

\keywords{Time dilation, clock drift, error correction, zero-error code, Riemann
zeta function.}

\pacs{89.70.Kn, 03.30.+p, 02.10.De}

\maketitle

\section{Introduction}

\vspace{-3mm}
Information theory \cite{shannon, csiszar+korner} -- one of the major
scientific achievements of the second half of the 20'th century -- was
developed by Shannon as a formal framework for the study of transmission
and processing of information in the classical domain.
In this paper, we introduce and analyze a problem which brings information
theory in the relativistic context and is, in particular, meant to illustrate
the effects of time dilation \cite{einstein, mermin} on the ultimate limits
of information transfer.

\vspace{-3mm}
\subsection{Model description}

\vspace{-3mm}
Consider the following model of communication between two parties:
Alice triggers $ k $ events at moments $ \tilde{t}_1, \tilde{t}_2, \ldots, \tilde{t}_k $,
which are selected from the set of integers $ \{1, 2, \ldots, N\} $ according
to a clock in her reference frame, that is $ \tilde{t}_i \in \{1, 2, \ldots, N\} $,
$ 1 \leqslant \tilde{t}_1 < \tilde{t}_2 < \cdots < \tilde{t}_k \leqslant N $,
and they are detected by Bob at moments
$ \alpha \tilde{t}_1, \alpha \tilde{t}_2, \ldots, \alpha \tilde{t}_k $
according to his own clock.
The factor $ \alpha $, modeling the difference in clock rates in the two
frames of reference, is not a priori known by either side.
(The two parties are assumed to be synchronized in the sense that they
have agreed on the beginning of time; this may be achieved by triggering
an additional event at time $ \tilde{t} = 0 $ in Alice's frame, upon
detecting of which Bob will set his own clock to $ \tilde{t} = 0 $ as well.)
It will be convenient to describe the signals by specifying the intervals
between consecutive events, rather than the time elapsed from the moment
$ \tilde{t} = 0 $ to each of the events.
In this notation, the set of all possible signals that may be transmitted
by Alice is represented by:
\begin{equation}
  T_{N, k}  =
	\left\{ (t_1, \ldots, t_k) \in \mathbb{N}^k  :  \sum_{i=1}^k t_i \leqslant N \right\} ,
\end{equation}
where $ \mathbb{N} = \{1, 2, \ldots \} $ denotes the set of positive integers.

Recovering the ``transmitted $ k $-tuple'' $ \bs{t}  =  ( t_1, t_2, \ldots, t_k ) $
from the ``received $ k $-tuple'' $ \alpha \bs{t} $ is trivial if the time
dilation factor $ \alpha $ is known and fixed, and, since the $ k $-tuple
$ \bs{t} $ is selected arbitrarily from among $ \binom{N}{k} $ possibilities,
$ \log_2\!\binom{N}{k} $ bits of information can be conveyed to Bob in this
way.
However, under our assumption that $ \alpha $ is not a priori known, not
all choices of $ \bs{t} $ can be used for reliable communication as some
of them are indistinguishable by the receiver.
For example, for $ k = 2 $, if Bob were to detect the two events at the
moments $ (2.1, 4.2) $, it would be impossible for him to deduce with
certainty whether $ \bs{t} = (1, 2) $ or $ \bs{t} = (2, 4) $.
In this case, the two parties have to agree beforehand to restrict the set
of allowed signals $ \bs{t} $ to a proper subset of $ T_{N,k} $ (a code)
in order to enable Bob to always infer the transmitted $ k $-tuple correctly,
regardless of $ \alpha $.
In other words, in the above scenario time dilation represents signal
distortion, and the following question then naturally arises:
how many bits of information can be reliably conveyed to the receiver
by using the simple form of communication just described?%

We adopt an information-theoretic approach and model the signal distortion
-- the unknown factor $ \alpha $ --  as an absolutely continuous random
variable with probability density function $ p_{\alpha}(\cdot) $ whose
support is either $ [1, \myind] $, for some constant $ \myind \in (1, \infty) $,
or $ [1, \infty) $.
The seemingly more general model where the support of $ p_{\alpha}(\cdot) $
is $ \big[ \underline{\alpha}, \myind \big] $, is equivalent to the case
$ \big[ 1, \myind / \underline{\alpha} \big] $ from the communication
viewpoint, so the lower end of the interval may be taken to be $ 1 $
without loss of generality.

\begin{remark}
\textnormal{
The difference in clock rates that is assumed in the model may be caused
by various physical effects, such as relative motion between Alice and Bob,
difference in the gravitational potential between them, clock imperfections,
etc.
We note that our results apply to all models with a linear change of time
scale $ \bs{t} \mapsto \alpha \bs{t} $ and do not depend on the physical
reasons causing this change.%
}%
\end{remark}

\vspace{-3mm}
\subsection{Zero-error codes}

\vspace{-3mm}
We say that two input signals $ \bs{t}', \bs{t}'' \in T_{N, k} $ are
\emph{confusable} if the receiver can confuse one with the other with
positive probability, meaning that the corresponding sets of output signals
$ \{ \alpha' \bs{t}' \} $ and $ \{ \alpha'' \bs{t}'' \} $
have infinite intersection (here $ \alpha' $ and $ \alpha'' $ vary through
the support of $ p_{\alpha}(\cdot) $).
A subset $ S \subseteq T_{N,k} $ is called a \emph{zero-error code}
\cite{shannon_ze} if any two distinct elements of $ S $ are non-confusable.
Elements of a code are called codewords.
Hence, a zero-error code is a set of permissible signals that can be
unambiguously distinguished by the receiver.
In other words, based on the observed signal, Bob will be able to identify
the codeword which produced that signal, and thereby recover the transmitted
information, with probability $ 1 $.

A zero-error code $ S \subseteq T_{N,k} $ is said to be optimal if it has
the largest cardinality among all such codes in $ T_{N,k} $.
This maximum cardinality is in general difficult to determine exactly for
arbitrary parameters $ N, k $, and it is instructive to focus on its asymptotic
behavior.
To that end, we define the maximum asymptotic \emph{density} of codes in the
space of $ k $-event signals:%
\begin{equation}
\label{eq:deltadef}
  \delta_{\myind}(k)
	 = \lim_{N \to \infty} \max_{S \subseteq T_{N, k}} \frac{|S|}{\binom{N}{k}} ,
\end{equation}
where the maximum is taken over all zero-error codes $ S \subseteq T_{N, k} $,
and where $ \myind \in (1, \infty] $ is the supremum of the support of
$ p_{\alpha}(\cdot) $.
(It will be evident from the analysis that the cardinality of optimal codes
depends on $ p_{\alpha}(\cdot) $ only through $ \myind $, justifying the
notation in \eqref{eq:deltadef}.)
Our aim is to characterize the quantity $ \delta_{\myind}(k) $ for every
$ k $ and $ \myind $.

\vspace{-3mm}
\section{Optimal signal sets and their density}

\vspace{-3mm}
We first consider the case of unbounded indeterminacy of the time dilation
factor, meaning that the support of the probability density function
$ p_{\alpha}(\cdot) $ is the entire half-line $ [1, \infty) $.
Time dilation distorts the signal represented by the point $ \bs{t} \in T(N, k) $
by multiplying this point by a random factor $ \alpha $, or equivalently,
by moving the point $ \bs{t} $ by a random amount along its ``line of sight''
from the origin; see Fig.~\ref{fig:code}.
Since any two points lying on the same line of sight from the origin are
confusable (because $ \operatorname{Pr}\{\alpha \geq \alpha_0\} > 0 $ for
any fixed $ \alpha_0 $), no two such points can belong to the same zero-error
code.
Therefore, an optimal code is obtained by selecting exactly one point on
each line of sight, and the simplest choice is to select the first point
encountered on each line.
The points $ \bs{t} $ in the grid $ \mathbb{N}^k $ that are encountered
first when going along the straight lines from the origin are those that
satisfy the condition $ \gcd(\bs{t}) \equiv \gcd(t_1, \ldots, t_k) = 1 $,
where $ \gcd $ denotes the greatest common divisor.
This shows that an optimal zero-error code for this model is:
\begin{equation}
\label{eq:code}
  C_{N,k}  =  \Big\{ \bs{t} \in T_{N,k} : \gcd(\bs{t}) = 1 \Big\} .
\end{equation}

\begin{figure}
\centering
  \includegraphics[width=0.95\columnwidth]{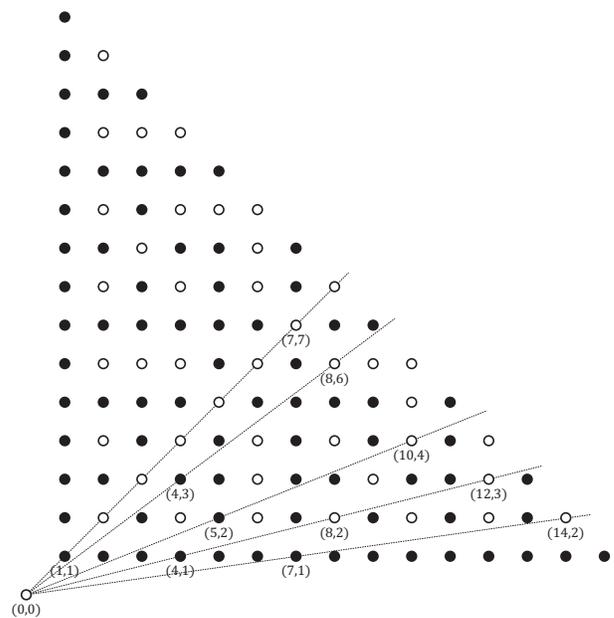}
\caption{The space $ T_{16,2} $ representing the signals consisting of
$ k = 2 $ events, each occurring in one of $ N = 16 $ time slots, and
the code $ C_{16,2} $ consisting of all the points $ \bs{t} = (t_1, t_2) \in T_{16,2} $
satisfying $ \gcd(t_1, t_2) = 1 $.
Codewords of $ C_{16,2} $ are represented as black dots.
Several ``lines of sight'' from the origin are also illustrated.}%
\label{fig:code}
\end{figure}%

In the case when the support of $ p_{\alpha}(\cdot) $ is a finite interval
$ [1, \myind] $, a code can be constructed by applying the following greedy
procedure on every line of sight from the origin:
select the first point $ \bs{t} $ on the line (the one satisfying
$ \gcd(\bs{t}) = 1 $) as a codeword, and exclude the points in
$ \big\{ \alpha \bs{t} : \alpha \in [1, \myind) \big\} $ which are confusable
with $ \bs{t} $;
then select the next point encountered on the line, $ \lceil \myind \rceil \bs{t} $,
as a codeword, and exclude the points in
$ \big\{ \alpha \lceil \myind \rceil \bs{t} : \alpha \in [1, \myind) \big\} $;
and so on.
This iterative procedure results in the following set:
\begin{equation}
\begin{aligned}
\label{eq:code_d}
  D_{N, k}
   = \left( \bigcupdot_{n=1}^{\infty}  d_{n} C_{N, k}  \right)  \cap  T_{N,k} ,
\end{aligned}
\end{equation}
where $ (d_{n})_{n=1}^{\infty} $ is the sequence defined by the following
recursion:
\begin{equation}
\label{eq:d}
  d_{n} = \lceil \myind d_{n-1} \rceil , \qquad
  d_{1} = 1 ,
\end{equation}
$ d_{n} C_{N, k} $ stands for
$ \big\{ d_{n} \bs{t} : \bs{t} \in C_{N, k} \big\} $, and
the notation $ \cupdot $ emphasizes that the sets in the union are pairwise
disjoint.
For reasons of simplicity, we do not make the dependence of $ d_n $ on
$ \myind $ explicit in the notation.
Note that, for $ \myind \in \mathbb{N} $, $ \myind \geqslant 2 $,
\eqref{eq:d} simplifies to:
\begin{equation}
\label{eq:dint}
  d_{n} = \myind^{n-1} .
\end{equation}
It is evident from the construction that $ D_{N, k} $ is indeed a zero-error
code: the only way two different codewords $ d_n \bs{t} $ and $ d_{n+1} \bs{t} $
can produce the same signal at the receiver's end is if $ \alpha $ takes the
value $ \myind $ (and if $ \myind \in \mathbb{N} $), which is a zero-probability
event.
We demonstrate in the following theorem that the code $ D_{N, k} $ is in fact
optimal and, based on this observation, we obtain a characterization of the
maximum asymptotic density $ \delta_{\myind}(k) $.
In order for the statement to be valid for the trivial case $ \myind = 1 $
as well, we define $ d_n = n $ for $ \myind = 1 $, which is justified by
a continuity argument (taking the limit $ \myind \to 1 $ in \eqref{eq:d}).

Recall the definition of the Riemann zeta function \cite{riemann}:%
\begin{equation}
\label{eq:zeta}
  \zeta(k)  =  \sum_{n=1}^{\infty} n^{-k} .
\end{equation}

\begin{theorem}
\label{thm:main}
Fix $ k \in \mathbb{N} $, $ k \geq 2 $, and $ \myind \in [1, \infty) $, and
let $ (d_{n})_{n=1}^{\infty} $ be the sequence defined in \eqref{eq:d}.
We then have:
\begin{align}
\label{eq:deltad}
  \delta_{\myind}(k)  =  \zeta(k)^{-1}  \sum_{n=1}^{\infty}  d_{n}^{-k} .
\end{align}
For $ \myind = \infty $, we have:
\begin{equation}
\label{eq:delta}
  \delta_{\infty}(k)  =  \zeta(k)^{-1} .
\end{equation}
\end{theorem}
\begin{proof}
Optimality of the code $ D_{N,k} $ follows directly from the result of
Shannon \cite[Thm 3]{shannon_ze}, which states that a zero-error code
$ S \subseteq T $ is optimal if there exists a mapping $ f : T \to S $
with the property that $ f(\bs{t}') \neq f(\bs{t}'') $ for any two non-confusable
$ \bs{t}', \bs{t}'' \in T $.
The required function $ f : T_{N,k} \to D_{N,k} $ in our case is defined
by mapping all the points in the set
$ \big\{ \alpha d_n \bs{t} : \alpha \in [1, \myind) \big\} \cap T_{N,k} $ to
$ d_n \bs{t} $, for every $ \bs{t} \in T_{N,k} $ with $ \gcd(\bs{t}) = 1 $,
and every $ n \geqslant 1 $.
Since any two non-confusable points in $ T_{N,k} $ belong to different sets
of the form $ \big\{ \alpha d_n \bs{t} : \alpha \in [1, \myind) \big\} $
(i.e., either $ n $ is different, or $ \bs{t} $ is, or both), they have
different images under $ f $.
Therefore, it follows from the quoted result of Shannon that $ D_{N,k} $ is
optimal, implying that:
\begin{equation}
\label{eq:deltak}
  \delta_{\myind}(k)  =  \lim_{N\to\infty} \frac{|D_{N,k}|}{\binom{N}{k}} .
\end{equation}
Now consider the relation \eqref{eq:code_d} and note that the subcodes
$ \big( d_{n} C_{N, k} \big) \cap T_{N, k} $ can be written as:
\begin{equation}
\begin{aligned}
  \big( d_{n} C_{N, k} \big) \cap T_{N, k}
	 &= \Big\{ \bs{t} \in T_{N, k} : \gcd(\bs{t}) = d_{n} \Big\}  \\
   &= d_{n} C_{\lfloor N/d_{n} \rfloor, k} .
\end{aligned}
\end{equation}
Their asymptotic density therefore equals:
\begin{equation}
\begin{aligned}
	\lim_{N \to \infty}  \frac{ \big| \big( d_{n} C_{N, k} \big) \cap T_{N, k} \big| }{ |T_{N, k}| }
  &= \lim_{N \to \infty}  \frac{ \big| C_{\lfloor N/d_{n} \rfloor, k} \big| }{ \binom{N}{k} }  \\
	&= \lim_{N' \to \infty}  \frac{ | C_{N', k} | }{ \binom{d_{n} N'}{k} }  \\
	&= \lim_{N' \to \infty}  \frac{ | C_{N', k} | }{ d_{n}^k \binom{N'}{k} }  \\
\label{eq:subcode}
	&= d_{n}^{-k}  \delta_{\infty}(k) ,
\end{aligned}
\end{equation}
where we have used the fact that the code $ C_{N,k} $ from \eqref{eq:code}
is optimal for the case $ \myind = \infty $ and hence:
\begin{align}
\label{eq:densinf}
  \delta_{\infty}(k) = \lim_{N \to \infty}  \frac{ | C_{N, k} | }{ \binom{N}{k} } .
\end{align}
Now \eqref{eq:code_d}, \eqref{eq:deltak} and \eqref{eq:subcode} imply:
\begin{equation}
\label{eq:deltad1}
  \delta_{\myind}(k)  =  \delta_{\infty}(k)  \sum_{n=1}^{\infty}  d_{n}^{-k} .
\end{equation}
Proving \eqref{eq:deltad} is thus reduced to proving \eqref{eq:delta}.
The density $ \delta_{\infty}(k) $ can be determined from \eqref{eq:code},
\eqref{eq:densinf}, and the fact that the probability of $ k $ random
positive integers being relatively prime equals $ \zeta(k)^{-1} $
\cite{nymann}, but we give here a direct derivation as well.
To this end, note that $ \delta_{1}(k) = 1 $ for every $ k $.
This is because the condition $ \myind = 1 $ means that the time dilation
factor is known exactly at the receiver, and hence the optimal code for
this case is trivially the set of all possible input signals, $ T_{N, k} $,
which has density $ 1 $.
We then get from \eqref{eq:deltad1} and the fact that $ d_{n} = n $
when $ \myind = 1 $:
\begin{equation}
\begin{aligned}
  1  
     =  \delta_{\infty}(k)  \sum_{n=1}^{\infty} n^{-k}
     =  \delta_{\infty}(k)  \zeta(k) ,
\end{aligned}
\end{equation}
which completes the proof of the theorem.
\end{proof}

For integral values of $ \myind $, the resulting density can be expressed
explicitly due to \eqref{eq:dint}.
Namely, for $ \myind \in \mathbb{N} $, $ \myind \geqslant 2 $:%
\begin{equation}
  \delta_{\myind}(k)  =  \frac{\myind^{k}}{\zeta(k) \big(\myind^{k}-1\big)} .
\end{equation}

As we have seen, when nothing is known about the factor $ \alpha $, the
relation \eqref{eq:delta} is obtained as an important special case of
\eqref{eq:deltad} (or, rather, the limiting case as $ \myind \to \infty $).
In particular, the largest asymptotic density of a set of $ 2 $-event
signals that are distinguishable by a receiver experiencing time dilation is:%
\begin{equation}
  \delta_{\infty}(2) = \frac{6}{\pi^2} .
\end{equation}
More generally, for any even $ k = 2m $:
\begin{equation}
  \delta_{\infty}(2m) = \frac{ (-1)^{m+1} 2 (2m)! }{ (2 \pi)^{2m} B_{2m} } ,
\end{equation}
where $ B_k $ are the Bernoulli numbers \cite[Sec.~1.5]{edwards}.

\begin{remark}
\textnormal{
A model closely related to the one presented in this paper, where, instead
of relativistic time dilation, signals are distorted by a synchronization
error known as clock drift, was analyzed in \cite{yeung, shaviv}.
Optimal zero-error codes for this model were determined in \cite{shaviv},
though no estimate of their cardinality was given.
It should be mentioned that the notion of ``confusability'' was defined in
\cite{shaviv} by the requirement that the signals in question cannot produce
the same signal at the receiver's end.
This definition is slightly different from ours which requires that the two
input signals causing confusion at the receiver should be a zero-probability
event.
The consequent characterization of optimal codes in \cite{shaviv} is
similar to \eqref{eq:code_d}, but with the sequence $ (b_n)_{n=1}^{\infty} $
in place of $ (d_n)_{n=1}^{\infty} $, where:
\begin{equation}
\label{eq:b}
  b_{n} = \lfloor \myind b_{n-1} + 1 \rfloor , \qquad
  b_{1} = 1 .
\end{equation}
In particular, for $ \myind \in \mathbb{N} $:
\begin{equation}
\label{eq:bint}
  b_{n} = 1 + \myind + \myind^2 + \cdots + \myind^{n-1} = \frac{ \myind^n - 1 }{ \myind - 1 } .
\end{equation}
It can be shown, in the same way as in the proof of Theorem \ref{thm:main},
that the largest asymptotic density of codes defined as in \cite{shaviv} equals:
\begin{align}
\label{eq:deltab}
  \zeta(k)^{-1}  \sum_{n=1}^{\infty}  b_{n}^{-k} .
\end{align}
Notice that $ b_n \geqslant d_n $ for every $ n \geqslant 1 $.
In particular, for $ \myind \in \mathbb{N} $, $ \myind \geqslant 2 $,
the inequality is strict for every $ n \geqslant 2 $, so the density
in \eqref{eq:deltab} is strictly smaller than the density $ \delta_{\myind}(k) $
from \eqref{eq:deltad}.
However, this is not always the case; for irrational $ \myind $, we have
$ b_n = d_n $ for every $ n \geqslant 1 $, and so the two densities are equal.
}
\end{remark}

\vspace{-3mm}
\section{Closing remarks}

\vspace{-3mm}
Quantifying information content and determining the fundamental limits of
information transmission are two of the main directions of study in information
theory.
The lines of research that explore such questions in various physical systems
have a long history in science, most notably in quantum information theory.
In this paper, we have described a scenario where information transfer is
considered in the relativistic context, and we have presented a result that
quantifies the limits of communication in this model.
The obtained solution, though simple, is interesting in that it provides a
link between Shannon's information theory, special relativity, and number
theory.
As such, it calls for further investigation of this and related problems, in
particular of the much more difficult case of non-linear change of time scale,
which may be caused by, e.g., accelerated observers.

\vspace{-3mm}
\section{Acknowledgment}

\vspace{-3mm}
This work was supported by the European Union's Horizon 2020 research
and innovation programme under Grant Agreement number 856967, and
by the Ministry of Education, Science and Technological Development of
the Republic of Serbia through the project number 451-03-68/2020-14/200156.

\end{document}